\documentclass{article}

\usepackage[]{graphicx,color}
\usepackage{amsmath,amssymb,amsthm}
\usepackage[top=30truemm,bottom=30truemm,left=25truemm,right=25truemm]{geometry}
\usepackage{authblk}

\newtheorem{theorem}{Theorem}[section]
\newtheorem{lemma}[theorem]{Lemma}

\newtheorem{definition}[theorem]{Definition}

\usepackage{bm}

\newcommand{\argmax}{\mathop\mathrm{arg\,max}}

\newcommand{\conv}{\mathop\mathrm{conv}}

\newcommand{\ALG}{\mathrm{ALG}}
\newcommand{\ALGs}{\mathcal{A}}
\newcommand{\RALGs}{\overline{\mathcal{A}}}
\newcommand{\OPT}{\mathrm{OPT}}
\newcommand{\CR}{\mathrm{CR}}
\newcommand{\kCR}{\overline{\CR}}

\newcommand{\dep}{\mathrm{dep}}
\renewcommand{\det}{\mathrm{det}}
\newcommand{\iid}{\mathrm{i.i.d.}}
\newcommand{\ind}{\mathrm{ind}}
\newcommand{\rd}{\mathrm{r.d.}}
\newcommand{\ri}{\mathrm{r.i.}}
\newcommand{\model}{\mathrm{model}}

\usepackage{tikz}
\usetikzlibrary{arrows,positioning,decorations.markings}

\title{Stochastic Input Models in Online Computing}
\author{Yasushi Kawase\thanks{Supported by a Grant-in-Aid for Young Scientists (B) (No.~16K16005).}}
\affil{Tokyo Institute of Technology, Tokyo, Japan\\
\texttt{kawase.y.ab@m.titech.ac.jp}}
\date{}

\begin{document}
\maketitle
\begin{abstract}
  In this paper, we study twelve stochastic input models for online problems and reveal the relationships among the competitive ratios for the models.
  The competitive ratio is defined as the worst ratio between the expected optimal value and the expected profit of the solution obtained by the online algorithm where the input distribution is restricted according to the model.
  To handle a broad class of online problems, we use a framework called \emph{request-answer games} that is introduced by Ben-David et al.
  The stochastic input models consist of two types: known distribution and unknown distribution.
  For each type, we consider six classes of distributions:
  dependent distributions, deterministic input, independent distributions, identical independent distribution, random order of a deterministic input, and random order of independent distributions.
  As an application of the models, we consider two basic online problems, which are variants of the secretary problem and the prophet inequality problem, under the twelve stochastic input models.
  We see the difference of the competitive ratios through these problems.
\end{abstract}

\section{Introduction}
In online computing, we are given a sequence of requests, and
we need to choose an action in each step based on the current information
without knowing the full information which will be completely obtained in the future~\cite{borodin2005,komm2016}.
Since an online algorithm is forced to make decisions without knowing the entire inputs, they may later turn out not to be optimal.
The quality of an online algorithm is usually measured by the \emph{competitive ratio}~\cite{sleator1985},
which is the worst ratio between the optimal value and the profit of the solution obtained by the online algorithm.
However, for practical application, it is too pessimistic because the worst-cases rarely occur in real-world.
Thus, an average-case analysis is more suitable than the worst one in such a situation.
In order to analyze average performance, we need to assume some distribution of inputs.
Hence, in this paper, we consider \emph{stochastic input models}.
In 2004, Hajiaghayi, Kleinberg, and Parkes~\cite{HKP2004} studied the case where requests are drawn independently from some unknown distribution
and the case where the requests sequence is determined by picking a multi-set of requests and then permuting them randomly.
However, we can define various stochastic input models, and we need to consider which should we use depending on analysis objective.

Various studies related to online problems with a stochastic input have been extensively conducted on statistics literature.
A most famous example of the problems is \emph{secretary problem}~\cite{ferguson1989,freeman1983,samuels1991}.
In the (classical) secretary problem,
a decision-maker is willing to hire the best secretary out of $n$ applicants that arrive in a random order,
and the goal is to maximize the probability of choosing the best applicant.
As each applicant appears, it must be either selected or rejected, and the decision is irrevocable.
It is assumed that the decision must be based only on the relative ranks (without ties) of the applicants seen so far and the number of applicants $n$.
It is well known that one can succeed with the optimal probability $1/e$ by the following algorithm:
observe the first $n/e$ applicants without selecting, and then select the next applicant who is the best among the observed applicants~\cite{dynkin1963toc}.

The above setting is also referred to the \emph{no-information} case
because there is no information about the value of applicants.
When the decision-maker is allowed to observe the actual values of the applicants,
the problem is called the \emph{full-information secretary problem} if the values are chosen independently and identically from a known distribution.
Also, if the values are chosen independently and identically from a distribution with an unknown parameter,
it is called the \emph{partial-information secretary problem}.
Moreover, when the values are chosen independently from known (not necessarily identical) distributions, the problem is studied as \emph{prophet inequality}.
Krengel, Sucheston, and Garling~\cite{KS1977} provided a tight $2$-competitive algorithm: select the first value that is at least half of the expected maximum.

There are also a huge number of researches on the stochastic input models in online learning literature,
such as online (convex) optimization, online prediction, online classification, and multi-armed bandit~\cite{CL2006,BC2012,hazan2016}.
In many studies in the literature, it is assumed that input is drawn from an unknown distribution and one repeatedly performs the same task.

The main purpose of this paper is to compare stochastic input models in the sense of competitive ratios.
As a previous result, Mehta~\cite{mehta2013} provided a relation between four models: unknown deterministic input, random order of an unknown deterministic input, unknown i.i.d., and known i.i.d.
However, he considered only the online allocation problems.
On the other hand, this paper studies a more general form of online problems through \emph{request-answer games}.

\subsection*{Our results}
In this paper, we provide twelve stochastic input models for online problems and reveal the relationships among the competitive ratios for the models.
The models consist of two types: known distribution and unknown distribution.
For each type, we consider six classes of distributions:
dependent distributions (dep), deterministic input (det), independent distributions (ind), identical independent distribution (i.i.d.), random order of a deterministic input (r.d.), and random order of independent distributions (r.i.).
The competitive ratio is defined as the worst ratio between the expected optimal value and the expected profit of the solution obtained by the online algorithm where the input (request) distribution is restricted according to the model.
We will denote by $\CR_{\model}$ (resp.\ $\kCR_{\model}$) the competitive ratio with unknown (resp.\ known) distribution $\model$.
For example, $\CR_{\det}$, $\CR_{\rd}$, $\CR_{\iid}$, $\kCR_{\iid}$ represent competitive ratios for unknown deterministic input, random order of an unknown deterministic input, unknown i.i.d., and known i.i.d., respectively.
Then, our results can be summarized as Figure~\ref{fig:relations}.

\begin{figure}[htbp]
\begin{center}
\begin{tikzpicture}[block/.style={rectangle, draw, text centered, rounded corners, text badly centered, text width=40, fill=gray!05}]
  \tikzset{->-/.style={decoration={markings, mark=at position .5 with {\arrow{>}}},postaction={decorate}}}
  \node[block] (kdet) {$\kCR_{\det}$};
  \node[block,right=of kdet] (kiid) {$\kCR_{\iid}$};
  \node[block,below right=of kdet] (kro) {$\kCR_{\rd}$};
  \node[block,below right=of kiid] (kps) {$\kCR_{\ri}$};
  \node[block,above right=10mm and 25mm of kiid] (kind) {$\kCR_{\ind}$};
  \node[block,right=of kiid] (iid) {$\CR_{\iid}$};
  \node[block,below right=2mm and 10mm of iid] (rops) {$\CR_{\rd}$\\[10pt] $\CR_{\ri}$};
  \node[block,above right=-10mm and 10mm of rops] (det) {$\kCR_{\dep}$\\[10pt] $\CR_{\det}$\\[10pt] $\CR_{\ind}$\\[10pt] $\CR_{\dep}$};

  \foreach \u / \v in {kdet/kiid, kdet/kro, kiid/iid, kro/kps, kiid/kind, kiid/kps, iid/rops, kind/det, kps/rops, rops/det}
  \draw[->-, thick] (\v) -- (\u);
\end{tikzpicture}
\end{center}
\caption{The relationships of the competitive ratios.
  Each arrow represents an inequality (the value at the head of each arrow is at most the one at its tail)
  and ratios in the same region have the same value.}\label{fig:relations}
\end{figure}
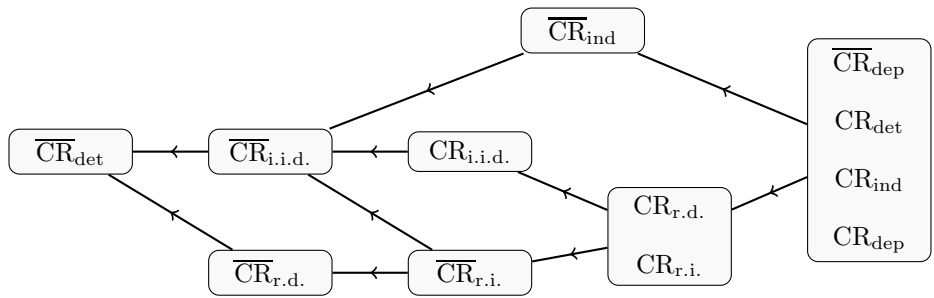

\subsection*{Related work}
Ben-David et al.~\cite{bendavid1994} introduced a general framework of online problems, which is called \emph{request-answer games}.
They compared oblivious, adaptive-online, and adaptive-offline adversaries.
The request-answer games can be seen as the following two-person zero-sum games.
The first player seeks an online algorithm that minimizes the competitive ratio;
the second player seeks a request sequence that maximizes the competitive ratio.
We can view a randomized online algorithm as a mixed strategy of the first player.
By using von Neumann's minimax theorem~\cite{vonneumann1928} and Loomis's lemma~\cite{loomis1946},
we can get a technique to obtain an upper bound of the competitive ratio.
This technique is called Yao's principle~\cite{yao1977} and used for various online problems.

Many generalizations of the secretary problem, i.e., problems in the random order model, have been studied in the competitive analysis literature.
As a natural generalization, Hajiaghayi, Kleinberg, and Parkes~\cite{HKP2004} introduced the multiple-choice secretary problem
and Kleinberg~\cite{kleinberg2005} provided an asymptotically optimal competitive algorithm when the objective is to maximize the sum of the $k$ choices.
Babaioff, Immorlica, and Kleinberg~\cite{BIK2007} introduced the matroid secretary problem.
In this model, the set of selected applicants must be an independent set in an underlying matroid.
Lachish~\cite{lachish2014} provided $O(\log\log r)$-competitive algorithm where $r$ is the rank of given matroid.
Korula and P\'{a}l~\cite{KP2009} presented 8-competitive algorithm for bipartite matching
and then Kesselheim et al.~\cite{KRTV2013} provide a tight $e$-competitive algorithm.
Babaioff et al.~\cite{BIKK2007} introduced knapsack secretary problem and proposed $10e$-competitive algorithm.
For more details, see Dinitz~\cite{dinitz2013} and Babaioff et al.~\cite{babaioff2008oaa}.

There are also a number of researches for generalizations of prophet inequalities.
Hajiaghayi, Kleinberg, and Sandholm~\cite{HKS2007} considered multiple-choice variant and
they provided a prophet inequality when the objective is to maximize the sum of the $k$ choices.
Kleinberg and Weinberg~\cite{KW2012} studied the matroid constrained version and gave a tight $2$-competitive algorithm.
Rubinstein and Singla~\cite{RS2017} introduced framework for combinatorial valuation functions.
Moreover, Esfandiari et al.~\cite{EHLM2015} introduced a prophet secretary problem, which is a natural combination of the prophet inequality problem and the secretary problem.
In their setting, the values are chosen independently from known (not necessarily identical) distributions but arrive in a random order.

\subsection*{Paper Organization}
Section~\ref{sec:model} describes formal definitions of our models via request-answer games.
Section~\ref{sec:example} discusses some variants of the secretary problem and the prophet inequality problem for our models.
In Section~\ref{sec:main}, we give the proof of our results shown in Figure~\ref{fig:relations}
and show that we cannot simplify the relations.

\section{Model}\label{sec:model}
In this section, we introduce request-answer games with stochastic input models.

A \emph{request-answer system} is a tuple \((R,A,f,n)\) where each component is defined as follows.
There is a set of \emph{requests} \(R\) and a set of \emph{answers} \(A\).
Throughout the paper, we assume that the sets $R$ and $A$ are finite.
A positive integer $n$ represents the number of rounds.
Let \(f:\,R^n\times A^n\to \mathbb{R}_+\) denote a \emph{utility function},
where $\mathbb{R}_+$ is the set of nonnegative real numbers.

A \emph{deterministic online algorithm} $\ALG$ is a sequence of functions \(g_i: R^i\to A\) \((i=1,2,\dots,n)\).
Given a deterministic online algorithm \(\ALG=(g_1,\dots,g_n)\) and 
a request sequence \(\bm{r}=(r_1,\dots,r_n)\in R^n\),
the output is an answer sequence $\ALG[\bm{r}]=(a_1,\dots,a_n)\in A^n$
where \(a_i=g_i(r_1,\dots,r_i)\) for \(i=1,\dots,n\).
The \emph{utility} incurred by \(\ALG\) on \(\bm{r}\), denoted by \(\ALG(\bm{r})\) is defined as
\begin{align*}
\ALG(\bm{r})=f(\bm{r},\ALG[\bm{r}]).
\end{align*}

A \emph{randomized online algorithm} $\ALG$ is defined as a probability distribution over 
the set of all deterministic online algorithms $\ALGs=\{(g_1,\dots,g_n)\mid g_i:R^i\to A~(i=1,\dots,n)\}$.
The \emph{utility} incurred by \(\ALG\) on \(\bm{r}\), denoted by \(\ALG(\bm{r})\) is defined as
\begin{align*}
\ALG(\bm{r})=\mathbb{E}_x[f(\bm{r},\ALG_x[\bm{r}])]
\end{align*}
where we use $\mathbb{E}_x$ to denote the expectation with respect to the distribution
over the set $\ALGs=\{\ALG_x\}$, which defines $\ALG$.
We will denote by $\RALGs$ the set of randomized algorithms.
We remark that $\RALGs$ is a convex compact set since $\ALGs$ is a finite set by the assumption that $R$ and $A$ are finite.

The performance of an online algorithm is measured
by the \emph{competitive ratio}---the ratio between its value and the optimal value for the worst request sequence.
Note that we only discuss in this paper the \emph{oblivious adversary}, i.e., the request sequence does not depend on the randomized results of the algorithm.
This is because we want to consider weaker adversaries than the standard model and
the adaptive adversary and stochastic input models are incompatible in most cases.

Given a request sequence \(\bm{r}\in R^n\), the \emph{optimal offline utility} on \(\bm{r}\) is defined as
\begin{align*}
\OPT(\bm{r})=\max\{f(\bm{r},\bm{a})\mid \bm{a}\in A^n\}.
\end{align*}

Let $\mathcal{S}_n$ be the set of permutations of $[n]$, where $[n]=\{1,\dots,n\}$.
For a permutation $\sigma\in\mathcal{S}_n$ and a sequence $\bm{x}=(x_1,\dots,x_n)$,
let us denote $\bm{x}^{\sigma}=(x_{\sigma(1)},\dots,x_{\sigma(n)})$.


\begin{definition}[Competitive Ratios]
For a request-answer game \((R,A,f,n)\),
a randomized online algorithm \(\ALG\in \RALGs\) is \(\rho\)-competitive for a distribution $D$ over the request sequences $R^n$ if
\begin{align*}
\mathbb{E}_{\bm{r}\sim D}[\OPT(\bm{r})]\le \rho\cdot \mathbb{E}_{\bm{r}\sim D}[\ALG(\bm{r})].
\end{align*}
Let $\mathcal{D}$ be a compact subset\footnotemark{} of the distributions over the request sequences $R^n$.
\footnotetext{We assume compactness to simplify the exposition.}

We define the competitive ratio of \((R,A,f,n)\) for \emph{unknown} distribution in $\mathcal{D}$ as
\begin{align}
\CR_{\mathcal{D}}(R,A,f,n)=\inf_{\ALG\in\RALGs}\sup_{D\in\mathcal{D}}\frac{\mathbb{E}_{\bm{r}\sim D}[\OPT(\bm{r})]}{\mathbb{E}_{\bm{r}\sim D}[\ALG(\bm{r})]} \label{eq:CR}
\end{align}
where we define $0/0=1$.
In other words, there exists a randomized online algorithm that is $\CR_{\mathcal{D}}(R,A,f,n)$-competitive for any $D\in\mathcal{D}$, and this is best possible.
Similarly,
we define the competitive ratio of \((R,A,f,n)\) for \emph{known} distribution in $\mathcal{D}$ as
\begin{align}
\kCR_{\mathcal{D}}(R,A,f,n)=\sup_{D\in\mathcal{D}}\inf_{\ALG\in\RALGs}\frac{\mathbb{E}_{\bm{r}\sim D}[\OPT(\bm{r})]}{\mathbb{E}_{\bm{r}\sim D}[\ALG(\bm{r})]}. \label{eq:kCR}
\end{align}
Namely, for any $D\in\mathcal{D}$, there exists a $\CR_{\mathcal{D}}(R,A,f,n)$-competitive algorithm, and this is best possible.
We often omit the argument $(R,A,f,n)$ if there is no confusion.
\end{definition}
The value of the competitive ratio is at least \(1\) and smaller is better.
It should be noted that, for minimization problems,
a randomized online algorithm \(\ALG\) is \(\rho\)-competitive for a distribution $D$ over the request sequences $R^n$ if
\begin{align*}
\mathbb{E}_{\bm{r}\sim D}[\ALG(\bm{r})]\le \rho\cdot \mathbb{E}_{\bm{r}\sim D}[\OPT(\bm{r})].
\end{align*}
We can also define the competitive ratio of a minimization version of request-answer games in the same manner.

By the max-min inequality, the following inequality holds.
\begin{lemma}\label{lemma:max-min-ineq}
For any request-answer game \((R,A,f,n)\) and a set of distributions $\mathcal{D}$ over the request sequences $R^n$, we have
\begin{align*}
  \CR_{\mathcal{D}}(R,A,f,n)\ge \kCR_{\mathcal{D}}(R,A,f,n).
\end{align*}
\end{lemma}
Moreover, if $\mathcal{D}$ is a compact convex set, the equality in the above inequality holds.

\begin{lemma}\label{lemma:minimax}
For any request-answer game \((R,A,f,n)\) and
a compact convex set of distributions $\mathcal{D}$ over the request sequences $R^n$, we have
\begin{align*}
  \CR_{\mathcal{D}}(R,A,f,n)=\kCR_{\mathcal{D}}(R,A,f,n).
\end{align*}
\end{lemma}
\begin{proof}
  Since $\CR_{\mathcal{D}}(R,A,f,n)\ge \kCR_{\mathcal{D}}(R,A,f,n)$ holds by Lemma~\ref{lemma:max-min-ineq},
  it is sufficient to prove $\CR_{\mathcal{D}}(R,A,f,n)\le \kCR_{\mathcal{D}}(R,A,f,n)$.
  Let $c=\CR_{\mathcal{D}}(R,A,f,n)$ and let $h(D,\ALG)=\mathbb{E}_{\bm{r}\sim D}[\OPT(\bm{r})-c\cdot\ALG(\bm{r})]$.
  Since $c=\CR_{\mathcal{D}}(R,A,f,n)=\inf_{\ALG\in\RALGs}\sup_{D\in\mathcal{D}}\frac{\mathbb{E}_{\bm{r}\sim D}[\OPT(\bm{r})]}{\mathbb{E}_{\bm{r}\sim D}[\ALG(\bm{r})]}$, it holds that for any $\ALG\in\RALGs$ there exists $D\in\mathcal{D}$ such that $h(D,\ALG)\ge 0$. Thus we have
  \begin{align*}
    \inf_{\ALG\in\RALGs}\sup_{D\in\mathcal{D}}h(D,\ALG)\ge 0.
  \end{align*}
  Moreover, by von Neumann's minimax theorem~\cite{vonneumann1928}, we have  
  \begin{align*}
    \inf_{\ALG\in\RALGs}\sup_{D\in\mathcal{D}}h(D,\ALG)=\sup_{D\in\mathcal{D}}\inf_{\ALG\in\RALGs}h(D,\ALG)
  \end{align*}
  since $\RALGs$ and $\mathcal{D}$ are compact convex sets, and
  $h(D,\ALG)$ is linear by the linearity of expectations.
  Thus, $\sup_{D\in\mathcal{D}}\inf_{\ALG\in\RALGs}h(D,\ALG)\ge 0$ and hence
  there exists $D^*\in\mathcal{D}$ such that $h(D^*,\ALG)\ge 0$ for any $\ALG\in\RALGs$.
  Therefore, we obtain
  \begin{align*}
    \kCR_{\mathcal{D}}(R,A,f,n)
    &=\sup_{D\in\mathcal{D}}\inf_{\ALG\in\RALGs}\frac{\mathbb{E}_{\bm{r}\sim D}[\OPT(\bm{r})]}{\mathbb{E}_{\bm{r}\sim D}[\ALG(\bm{r})]}\\
    &\ge \inf_{\ALG\in\RALGs}\frac{\mathbb{E}_{\bm{r}\sim D^*}[\OPT(\bm{r})]}{\mathbb{E}_{\bm{r}\sim D^*}[\ALG(\bm{r})]}\\
    &\ge c=\CR_{\mathcal{D}}(R,A,f,n),
  \end{align*}  
  which proves the lemma.
\end{proof}

A distribution over a (finite) domain $\Omega$ is a function $D:\Omega\to[0,1]$ such that $\sum_{x\in\Omega}D(x)=1$.
We denote by $\Delta(\Omega)$ the set of all such distributions.
For $x\in\Omega$, let $\delta_{x}$ be a distribution such that $\delta_{x}(x)=1$ and $\delta_{x}(y)=0$ for any $x,y\in\Omega$ such that $y\ne x$.
In this paper, we consider the following six classes of distributions.
\begin{enumerate}
\item \emph{Dependent distribution}:
  $$\Delta_{\dep}(R^n)=\Delta(R^n).$$
\item \emph{Deterministic input}:
  $$\Delta_{\det}(R^n)=\{\delta_{\bm{r}}\mid \bm{r}\in R^n\}.$$
\item \emph{Independent distribution}:
  $$\Delta_{\ind}(R^n)=\left\{D\in\Delta(R^n)\,\middle|\,D(\bm{r})=\prod_{i=1}^n D_i(r_i),~D_i\in\Delta(R)\right\}.$$
\item \emph{Identical independent distribution}:
  $$\Delta_{\iid}(R^n)=\left\{D\in\Delta(R^n)\,\middle|\,D(\bm{r})=\prod_{i=1}^n D'(r_i),~D'\in\Delta(R)\right\}.$$
\item \emph{Random order of deterministic input}:
  $$\Delta_{\rd}(R^n)=\left\{\psi_{\bm{r}}\middle|\,\bm{r}\in R^n\right\}$$
  where
  $\psi_{\bm{r}}\in\Delta(R^n)$ $(\bm{r}\in R^n)$ is a distribution such that
  $\psi_{\bm{r}}(\bm{r}')=|\{\sigma\in\mathcal{S}_n\mid \bm{r}'=\bm{r}^{\sigma}\}|/n!$ for $\bm{r}'\in R^n$.
  Here, recall that $\bm{r}^{\sigma}=(r_{\sigma(1)},\dots,r_{\sigma(n)})$.
\item \emph{Random order of independent distribution}:
  $$\Delta_{\ri}(R^n)=\left\{D\in\Delta(R^n)\,\middle|\,D(\bm{r})=\sum_{\sigma\in\mathcal{S}_n}\prod_{i=1}^n D_{\sigma(i)}(r_i)/n!,~D_i\in\Delta(R)\right\}.$$
\end{enumerate}
We will omit the argument $R^n$ when no confusion can arise.
In addition, we abbreviate $\CR_{\Delta_{\dep}(R^n)}(R,A,f,n)$ to $\CR_{\dep}$
and also we abbreviate the other competitive ratios in the same way.

Now, we can formally state our main results, which was depicted in Figure~\ref{fig:relations}.
\begin{theorem}\label{thm:main}
For any request-answer game $(R,A,f,n)$, we have
\begin{align}
&\kCR_{\det}\le \kCR_{\iid}\le \CR_{\iid}\le \CR_{\rd}=\CR_{\ri}\le \kCR_{\dep}=\CR_{\det}=\CR_{\ind}=\CR_{\dep},\label{eq:main1}\\
&\kCR_{\det}\le \kCR_{\rd}\le \kCR_{\ri}\le \CR_{\ri},\label{eq:main2}\\
&\kCR_{\iid}\le \kCR_{\ri},\label{eq:main3}\\
&\kCR_{\iid}\le \kCR_{\ind}\le \kCR_{\dep}.\label{eq:main4}
\end{align}
\end{theorem}

\section{Examples of the Stochastic Input Models}\label{sec:example}
Before proving our main results, we see some examples of the stochastic input models.

\subsection{Online selection problem}\label{subsec:osp}
  Suppose that a decision-maker sequentially observes a sequence of random variables $X_1,X_2,\allowbreak\dots,X_n$
  and is allowed to choose only one number, which can be done only upon receiving that number.
  The goal is to maximize the expectation of the chosen value.
  We formalize the problem as a request-answer game.
  Let $R_m=\{0,1,\dots,m\}$ (where $m$ is a sufficiently large integer), $A=\{0,1\}$, and 
  \begin{align*}
    f(\bm{r},\bm{a})=
    \begin{cases}
      r_{i^*} & (\text{if }\{i^*\}=\{i\mid a_i=1\}),\\
      0 & (\text{otherwise}).
    \end{cases}
  \end{align*}
  $a_i=1$ (resp.\ $a_i=0$) represents that the decision-maker selects (resp.\ rejects) $i$th value.

  The problem in the known independent distributions model is introduced by Krengel and Sucheston,
  and it is eagerly studied under the name of \emph{the prophet inequalities}.
  It is well known that this problem is $2$-competitive when $n\ge 2$, i.e., $\limsup_{m\to\infty} \kCR_{\Delta_{\ind}(R_m^n)}=2$ for any $n\ge 2$~\cite{HK1992}.
  Hill and Kertz~\cite{HK1982} gave a better upper bound of the competitive ratio in the known i.i.d.\ model
  and the current best upper bound is $1.354$~\cite{AEEHKL2017}.

  For the known dependent distributions model,
  Hill and Kertz~\cite{HK1983} proved that there exists no constant competitive algorithm, i.e., $\limsup_{n,m\to\infty} \kCR_{\Delta_{\dep}(R_m^n)}=\infty$.
  By Theorem~\ref{thm:main}, this implies that
  \begin{align*}
    \limsup_{n,m\to\infty} \CR_{\Delta_{\dep}(R_m^n)}=\limsup_{n,m\to\infty} \CR_{\Delta_{\ind}(R_m^n)}=\limsup_{n,m\to\infty} \CR_{\Delta_{\det}(R_m^n)}=\infty.
  \end{align*}

  For random order of known independent distributions model (prophet secretary problem), Esfandiari et al.~\cite{EHLM2015} provided $e/(e-1)$-competitive algorithm
  and there is no online algorithm that can achieve a competitive ratio better than $15/11$.
  Hence, it holds that $15/11\le \limsup_{n,m\to\infty}\kCR_{\Delta_{\ri}(R_m^n)}\le e/(e-1)$.

  When the given input is random order or i.i.d.~distribution,
  it is easy to see that the secretary algorithm is $e$-competitive.
  On the other hand, Stewart~\cite{Stewart1978} proves that no algorithm can select the best number with probability $1/e-o(1)$
  when $X_1,X_2,\dots,X_n$ are chosen i.i.d.~from a uniform distribution on the interval $(\alpha,\beta)$
  and $(\alpha,\beta)$ is chosen from a certain Pareto distribution.
  Transforming the variables $X_i$ to $M^{X_i}$, where $M$ is a sufficient large number,
  implies that there is no online algorithm that is better than $e$-competitive.  
  Thus we have $\limsup_{n,m\to\infty}\CR_{\Delta_{\rd}(R_m^n)}=\limsup_{n,m\to\infty}\CR_{\Delta_{\ri}(R_m^n)}=\limsup_{n,m\to\infty}\CR_{\Delta_{\iid}(R_m^n)}=e$.

  For the other models, we can easily check that $\kCR_{\Delta_{\rd}(R_m^n)}=\kCR_{\Delta_{\det}(R_m^n)}=1$ for any $n,m$
  because we can deterministically pick the maximum number.

\subsection{Odds problem}\label{subsec:odds}
  We show that the \emph{odds problem}---a generalization of the classical secretary problem---can be formulated as the known identical distribution model.
  In the odds problem, a decision-maker sequentially observes a sequence of independent 0/1 random variables $X_1,X_2,\dots,X_n$,
  where $\Pr[X_i=1]=p_i(1)$ and $\Pr[X_i=0]=p_i(0)~(=1-p_i(1))$.
  We say \emph{success} if $X_i=1$ and \emph{failure} if $X_i=0$.
  The goal is to find an optimal stopping rule to maximize the probability of win---the probability of obtaining the last success.
  Note that the special case $p_i=1/i$ corresponds to the classical secretary problem
  and several generalizations of the odds problem have been studied~\cite{BP2000,Tamaki2010,MA2014,MA2016}.

  Let us consider a request-answer game $(R,A,f,n)$ where $R=\{0,1\}$, $A=\{0,1\}$, and
  \begin{align*}
    f(\bm{r},\bm{a})=
    \begin{cases}
      1 & (\text{if }\bm{r}=(0,\dots,0)\text{ or }\{i\mid a_i=1\}=\argmax\{i\mid r_i=1\}),\\
      0 & (\text{otherwise}).
    \end{cases}
  \end{align*}
  Then, for a distribution $D^*$ such that $D^*(\bm{r})=\prod_{i=1}^n p_i(r_i)$,
  the expected value $\mathbb{E}_{\bm{r}\sim D^*}[\ALG(\bm{r})]$ is the probability of win\footnote{More precisely, we add extra probability $\Pr[\bm{r}=(0,\dots,0)]$ to simplify later discussion.}.
  Note that $\OPT(\bm{r})=1$ for any $\bm{r}\in R^n$ and hence the competitive ratio is the inverse value of the probability of win.
  
  Bruss~\cite{bruss2000} proved that $\ALG^*=(g_1^*,\dots,g_n^*)$ is $e$-competitive when
  \begin{align*}
    g_i^*(r_1,\dots,r_i)=\begin{cases}
    1&(r_i=1\text{ and }\sum_{j=i+1}^n\frac{p_j(1)}{p_j(0)}<1),\\
    0&(\text{otherwise})
    \end{cases}
  \end{align*}
  and this is asymptotically best possible.
  Moreover, no online algorithm can achieve a competitive ratio better than $e$
  even if the input is a known i.i.d.\ distribution when $n$ goes to infinity (see Lemma~\ref{lemma:odds_iid} in Appendix).
  Namely, we can conclude that
  \begin{align*}
    \limsup_{n\to\infty}\kCR_{\ind}=\limsup_{n\to\infty}\kCR_{\iid}=e.
  \end{align*}

  Next, let us consider the following distribution $D^*$:
  \begin{align*}
    D^*(\bm{r})=\begin{cases}
    1/n &(\bm{r}=(\underbrace{1,\dots,1}_{i},\underbrace{0,\dots,0}_{n-i}),~i=1,2,\dots,n),\\
    0   &(\text{otherwise}).
    \end{cases}
  \end{align*}
  Then, it is easy to see that any algorithm wins with probability at most $1/n$.
  On the other hand, selecting a variable uniformly at random is an $n$-competitive algorithm.
  Thus, we obtain $\kCR_{\dep}=\CR_{\det}=\CR_{\ind}=\CR_{\dep}=n$.

  For known deterministic inputs or known random order distribution,
  it holds that $\kCR_{\det}=\kCR_{\rd}=1$
  because the algorithm can distinguish whether a current success is the last one or not.

  For the other models $\CR_{\iid}$, $\CR_{\rd}$, $\CR_{\ri}$, and $\kCR_{\ri}$,
  we can observe, by Theorem~\ref{thm:main}, that the values are at least $e$ when $n$ goes to infinity and they are at most $n$.
  However, any good bounds are not known, to the best of the author's knowledge.

\section{Relating the Competitive Ratios}\label{sec:main}
In this section, we give the proof of Theorem~\ref{thm:main} and
we show that the relations~\eqref{eq:main1}--\eqref{eq:main4} cannot be merged into one sequence.

\subsection{Proof of Theorem~\ref{thm:main}}
Let us start with an easy part.
\begin{lemma}
  For any request-answer game $(R,A,f,n)$, we have
  \begin{align*}
    \kCR_{\Delta_{\det}}(R^n)=1.
  \end{align*}
\end{lemma}
\begin{proof}
  As the online algorithm knows the whole request sequence in advance, it can output the optimal answer sequence.
  More precisely, we get
  \begin{align*}
    \kCR_{\det}
    &=\sup_{\delta_{\bm{r}}\in\Delta_{\det}(R^n)}\inf_{\ALG\in\RALGs}\frac{\mathbb{E}_{\bm{r}'\sim \delta_{\bm{r}}}[\OPT(\bm{r}')]}{\mathbb{E}_{\bm{r}'\sim \delta_{\bm{r}}}[\ALG(\bm{r}')]}\\
    &=\sup_{\bm{r}\in R^n}\inf_{\ALG\in\RALGs}\frac{\OPT(\bm{r})}{\ALG(\bm{r})}\\
    &=\sup_{\bm{r}\in R^n}\frac{\max_{\bm{a}\in A^n} f(\bm{r},\bm{a})}{\max_{\ALG\in\RALGs}\ALG(\bm{r})}\\
    &=\sup_{\bm{r}\in R^n}\frac{\max_{\bm{a}\in A^n} f(\bm{r},\bm{a})}{\max_{\bm{a}\in A^n} f(\bm{r},\bm{a})}=1.\qedhere
  \end{align*}
\end{proof}
As the competitive ratios are at least $1$,
the lemma implies that $\kCR_{\det}$ is the smallest one.

We use the convex hull of a set of distributions to prove the rest part of the theorem.
For a set of distributions $\mathcal{D}$, the convex hull is the set
\begin{align*}
\conv(\mathcal{D})=\left\{\sum_{D\in\mathcal{D}}\lambda_D\cdot D \,\middle|\, \sum_{D\in\mathcal{D}}\lambda_D=1\text{ and }\lambda_D\ge 0~(\forall D\in\mathcal{D})\right\}.
\end{align*}

The following lemma connects the unknown distribution model and the known distribution model.
\begin{lemma}\label{lemma:unknown_known}
  For any compact set of distributions $\mathcal{D}$, we have
  \begin{align*}
    \CR_{\mathcal{D}}=\CR_{\conv(\mathcal{D})}=\kCR_{\conv(\mathcal{D})}.
  \end{align*}
\end{lemma}
\begin{proof}
  $\CR_{\conv(\mathcal{D})}=\kCR_{\conv(\mathcal{D})}$ holds by Lemma~\ref{lemma:minimax}
  since $\conv(\mathcal{D})$ is a compact convex set.

  Since $\mathcal{D}\subseteq \conv(\mathcal{D})$, we have
  \begin{align}
    \inf_{\ALG\in\RALGs}\sup_{D\in\mathcal{D}}\frac{\mathbb{E}_{\bm{r}\sim D}[\OPT(\bm{r})]}{\mathbb{E}_{\bm{r}\sim D}[\ALG(\bm{r})]}\le
    \inf_{\ALG\in\RALGs}\sup_{D\in\conv(\mathcal{D})}\frac{\mathbb{E}_{\bm{r}\sim D}[\OPT(\bm{r})]}{\mathbb{E}_{\bm{r}\sim D}[\ALG(\bm{r})]}. \label{eq:conveq}
  \end{align}
  Thus, it is sufficient to prove the reverse inequality of~\eqref{eq:conveq}.
 
  Let us fix $\ALG\in\RALGs$.
  Let $D^1,D^2,\dots\in\conv(\mathcal{D})$ be a sequence of distributions satisfying
  \begin{align*}
    \lim_{i\to\infty}\frac{\mathbb{E}_{\bm{r}\sim D^i}[\OPT(\bm{r})]}{\mathbb{E}_{\bm{r}\sim D^i}[\ALG(\bm{r})]}
    =\sup_{D\in\conv(\mathcal{D})}\frac{\mathbb{E}_{\bm{r}\sim D}[\OPT(\bm{r})]}{\mathbb{E}_{\bm{r}\sim D}[\ALG(\bm{r})]}.
  \end{align*}
  By the definition of the convex closure, for each $i\in\{1,2,\dots\}$,there exist coefficients $\lambda^i_{D}$ such that
  $D^i = \sum_{D\in\mathcal{D}} \lambda^i_{D}\cdot D$,
  $\sum_{D\in\mathcal{D}}\lambda^i_{D}=1$, and
  $\lambda^i_{D}\ge 0$ $(\forall D\in\mathcal{D})$.
  Thus, we have
  \begin{align*}
    \frac{\mathbb{E}_{\bm{r}\sim D^i}[\OPT(\bm{r})]}{\mathbb{E}_{\bm{r}\sim D^i}[\ALG(\bm{r})]}
    &=\frac{\sum_{\bm{r}\in R^n} D^i(\bm{r})\cdot \OPT(\bm{r})}{\sum_{\bm{r}\in R^n} D^i(\bm{r})\cdot \ALG(\bm{r})}\\
    &=\frac{\sum_{\bm{r}\in R^n} \sum_{D\in\mathcal{D}} \lambda^i_{D}\cdot D(\bm{r})\cdot \OPT(\bm{r})}{\sum_{\bm{r}\in R^n} \sum_{D\in\mathcal{D}} \lambda^i_{D}\cdot D(\bm{r})\cdot \ALG(\bm{r})}\\
    &=\frac{\sum_{D\in\mathcal{D}} \lambda^i_{D}\cdot \mathbb{E}_{\bm{r}\sim D}[\OPT(\bm{r})]}{\sum_{D\in\mathcal{D}} \lambda^i_{D}\cdot \mathbb{E}_{\bm{r}\sim D}[\ALG(\bm{r})]}\\
    &\le \sup_{D\in\mathcal{D}}\frac{\mathbb{E}_{\bm{r}\sim D}[\OPT(\bm{r})]}{\mathbb{E}_{\bm{r}\sim D}[\ALG(\bm{r})]}.
  \end{align*}
  Hence we obtain
  \begin{align*}
    \sup_{D\in\conv(\mathcal{D})}\frac{\mathbb{E}_{\bm{r}\sim D}[\OPT(\bm{r})]}{\mathbb{E}_{\bm{r}\sim D}[\ALG(\bm{r})]}
    =\lim_{i\to\infty}\frac{\mathbb{E}_{\bm{r}\sim D^i}[\OPT(\bm{r})]}{\mathbb{E}_{\bm{r}\sim D^i}[\ALG(\bm{r})]}
    \le \sup_{D\in\mathcal{D}}\frac{\mathbb{E}_{\bm{r}\sim D}[\OPT(\bm{r})]}{\mathbb{E}_{\bm{r}\sim D}[\ALG(\bm{r})]}.
  \end{align*}
  By taking infimum over $\ALG\in\RALGs$, we have
  \begin{align*}
    \inf_{\ALG\in\RALGs}\sup_{D\in\conv(\mathcal{D})}\frac{\mathbb{E}_{\bm{r}\sim D}[\OPT(\bm{r})]}{\mathbb{E}_{\bm{r}\sim D}[\ALG(\bm{r})]}
    \le \inf_{\ALG\in\RALGs}\sup_{D\in\mathcal{D}}\frac{\mathbb{E}_{\bm{r}\sim D}[\OPT(\bm{r})]}{\mathbb{E}_{\bm{r}\sim D}[\ALG(\bm{r})]},
  \end{align*}
  which proves the lemma.  
\end{proof}

Next, we examine relationships between our intended sets of distributions.
\begin{lemma}\label{lemma:conv_rel}
For any request-answer game $(R,A,f,n)$, we have
\begin{align*}
    \conv(\Delta_{\iid})\subseteq \conv(\Delta_{\rd})=\conv(\Delta_{\ri}) \subseteq \conv(\Delta_{\det})=\conv(\Delta_{\ind})=\Delta_{\dep}=\conv(\Delta_{\dep}).
\end{align*}

\end{lemma}
\begin{proof}
  We first prove that $\conv(\Delta_{\det})=\conv(\Delta_{\ind})=\conv(\Delta_{\dep})=\Delta_{\dep}$.
  By the definition, we have
  \begin{align*}
    \Delta_{\det}\subseteq \Delta_{\ind}\subseteq \Delta_{\dep}.
  \end{align*}
  Also, for $D\in\Delta_{\dep}$, it holds that $D=\sum_{\bm{r}\in R^n}D(\bm{r})\delta_{\bm{r}}$
  and hence $\Delta_{\dep}\subseteq \conv(\Delta_{\det})$.
  Thus, $\Delta_{\dep}=\conv(\Delta_{\dep})=\conv(\Delta_{\ind})=\conv(\Delta_{\det})$ holds.

  Next, we observe that $\conv(\Delta_{\rd}(R^n)) \subseteq \Delta_{\dep}(R^n)$.
  As $\Delta_{\dep}(R^n)=\conv(\Delta_{\dep}(R^n))$, it is sufficient to prove that $\Delta_{\rd}(R^n)\subseteq \Delta_{\dep}(R^n)$.
  Let $\psi_{\bm{r}}\in\Delta_{\rd}(R^n)$. Recall that $\psi_{\bm{r}}(\bm{r}')=|\{\sigma\in\mathcal{S}_n\mid \bm{r}'=\bm{r}^{\sigma}\}|/n!$ for $\bm{r}'\in R^n$.
  Thus we obtain
  \begin{align*}
    \psi_{\bm{r}}=\sum_{\bm{r}'\in R^n}\frac{|\{\sigma\in\mathcal{S}_n\mid \bm{r}'=\bm{r}^{\sigma}\}|}{n!}\cdot \delta_{\bm{r}'}\in \Delta_{\dep}(R^n).
  \end{align*}

  Furthermore, we show $\conv(\Delta_{\rd})=\conv(\Delta_{\ri})$.
  Since $\Delta_{\rd}\subseteq \Delta_{\ri}$, it is sufficient to prove that $\Delta_{\ri}\subseteq \conv(\Delta_{\rd})$.
  Let $D\in\Delta_{\ri}$ where $D(\bm{r})=\sum_{\sigma\in \mathcal{S}_n}\prod_{i=1}^n D_{\sigma(i)}(r_i)/n!$ and $D_i\in\Delta(R)$.
  Also, let $D'\in\Delta_{\ind}$ be a distribution such that $D'(\bm{r})=\prod_{i=1}^n D_i(r_i)$.
  Then we have
  \begin{align*}
    D&=\sum_{\bm{r}'\in R^n} D'(\bm{r}')\psi_{\bm{r}'}\in \conv(\Delta_{\rd}).
  \end{align*}
  
  Finally, we show that $\Delta_{\iid}(R^n)\subseteq \conv(\Delta_{\rd}(R^n))$.
  Let $D\in\Delta_{\iid}(R^n)$ and $D(\bm{r})=\prod_{i=1}^n D'(r_i)$.
  Then we have $D(\bm{r})=\prod_{i=1}^n D'(r_{i})=\prod_{i=1}^n D'(r_{\sigma(i)})=D(\bm{r}^{\sigma})$ for any $\sigma\in\mathcal{S}_n$
  and hence
  \begin{align*}
    D
    &=\sum_{\bm{r}\in R^n} D(\bm{r})\delta_{\bm{r}}
    =\sum_{\sigma\in\mathcal{S}_n}\frac{1}{n!}\sum_{\bm{r}\in R^n} D(\bm{r}^{\sigma})\delta_{\bm{r}^{\sigma}}\\
    &=\sum_{\bm{r}\in R^n}\sum_{\sigma\in\mathcal{S}_n} \frac{D(\bm{r})\delta_{\bm{r}^{\sigma}}}{n!}
    =\sum_{\bm{r}\in R^n} D(\bm{r})\sum_{\sigma\in\mathcal{S}_n} \frac{\delta_{\bm{r}^{\sigma}}}{n!}\\
    &=\sum_{\bm{r}\in R^n} D(\bm{r})\left(\sum_{\bm{r}'\in R^n}\frac{|\{\sigma\in\mathcal{S}_n\mid \bm{r}'=\bm{r}^{\sigma}\}|}{n!}\cdot \delta_{\bm{r}'}\right)\\
    &=\sum_{\bm{r}\in R^n} D(\bm{r})\cdot \psi_{\bm{r}}
    \in\conv(\Delta_{\rd}(R^n)),
  \end{align*}
  which proves the claim.
\end{proof}
We remark that $\Delta_{\iid}(R^n)\subsetneq \conv(\Delta_{\iid}(R^n))$ in general.
For example, if $R=\{0,1\}$ and $n=2$, then
$\frac{1}{2}\delta_{(0,0)}+\frac{1}{2}\delta_{(1,1)}\in \conv(\Delta_{\iid}(R^n))$
while $\frac{1}{2}\delta_{(0,0)}+\frac{1}{2}\delta_{(1,1)}\not\in \Delta_{\iid}(R^n)$.

Now we are ready to prove Theorem~\ref{thm:main}.
By Lemmas~\ref{lemma:unknown_known} and~\ref{lemma:conv_rel}, we obtain
\begin{align*}
\CR_{\iid}\le \CR_{\rd}=\CR_{\ri}\le \kCR_{\dep}=\CR_{\det}=\CR_{\ind}=\CR_{\dep}
\end{align*}
since $\mathcal{D}\subseteq\mathcal{D}'$ implies $\CR_{\mathcal{D}}\le\CR_{\mathcal{D}'}$.
Moreover, $\kCR_{\iid}\le \CR_{\iid}$ holds by Lemma~\ref{lemma:max-min-ineq}
and $\kCR_{\iid}\ge 1~(=\kCR_{\det})$ by the definition of competitive ratio.
Thus, \eqref{eq:main1} holds.
Also, \eqref{eq:main2} holds by $\Delta_{\rd}\subseteq \Delta_{\ri}$,
\eqref{eq:main3} holds by $\Delta_{\iid}\subseteq \Delta_{\ri}$,
and~\eqref{eq:main4} holds by $\Delta_{\iid}\subseteq \Delta_{\ind}\subseteq \Delta_{\dep}$.

\subsection{Inequalities for incomparable ratios}
We show that the relations~\eqref{eq:main1}--\eqref{eq:main4} cannot be merged into one sequence.
To prove this, we observe some examples of request-answer games.
\begin{lemma}
  For each following statement, there exists a request-answer game $(R,A,f,n)$ that satisfies it:
  $(a)$ $\kCR_{\ri}<\CR_{\iid}$,
  $(b)$ $\kCR_{\ind}<\CR_{\iid}$,
  $(c)$ $\kCR_{\rd}<\kCR_{\iid}$,
  $(d)$ $\CR_{\iid}<\kCR_{\rd}$, and
  $(e)$ $\CR_{\rd}<\kCR_{\ind}$.
\end{lemma}
\begin{proof}
\noindent$\boldsymbol{(a)}$ $\kCR_{\ri}<\CR_{\iid}$.
This inequality holds for an online selection problem instance since $\kCR_{\ri}\le e/(e-1)~(\approx 1.582)$ and $\CR_{\iid}\ge e~(\approx 2.718)$ when $n$ and $m$  go to infinity (see Section~\ref{subsec:osp}).

\medskip
\noindent$\boldsymbol{(b)}$ $\kCR_{\ind}<\CR_{\iid}$.
This inequality holds for an online selection problem instance since $\kCR_{\ind}\le 2$ and $\CR_{\iid}\ge e$ when $n$ and $m$ go to infinity (see Section~\ref{subsec:osp}).

\medskip
\noindent$\boldsymbol{(c)}$ $\kCR_{\rd}<\kCR_{\iid}$.
This inequality holds for an odds problem instance since $\kCR_{\rd}\to 1$ and $\kCR_{\iid}\to e$ as $n$ goes to infinity (see Section~\ref{subsec:odds}).\\

To prove the other cases, we consider request-answer games of the form $G=(R,A,f_S,n)$
where $S\subseteq R^n$, $A=\{0,1\}$, and
\begin{align*}
  f_S(\bm{r},\bm{a})=
  \begin{cases}
    1 &((\bm{r}\in S\text{ and }a_1=1)\text{ or }(\bm{r}\not\in S\text{ and }a_1=0)),\\
    0 &(\text{otherwise}).
  \end{cases}
\end{align*}
By the definition of the game, the task of the online algorithm is to predict $\bm{r}\in S$ or not.
The prediction is answered as $a_1$.
The profit is one if the prediction is correct and is zero otherwise.
Here, we remark that $\OPT(\bm{r})=1$ for any $\bm{r}\in R^n$.

\medskip
\noindent$\boldsymbol{(d)}$ $\CR_{\iid}<\kCR_{\rd}$.
We consider the game $G$ with
\begin{align*}
  R=\{1,2,\dots,n\} \quad\text{and}\quad S=\{(\sigma(1),\dots,\sigma(n))\in R^n\mid \sigma\in\mathcal{S}_n,~\sigma(n)>\sigma(n-1)\}.
\end{align*}
Then, for any i.i.d.\ distribution, the request sequence does not belong to $S$ with high probability.
Thus, by predicting that the request sequence does not belong to $S$, one can obtain an expected profit $1-o(1)$.
On the other hand, when the input distribution is $\psi_{(1,2,\dots,n)}$,
we cannot predict that the request sequence belongs to $S$ or not with probability better than $1/2$.
Thus, $\CR_{\iid}=1+o(1)$ and $\kCR_{\rd}=2$ and hence the inequality holds.

\medskip
\noindent$\boldsymbol{(e)}$ $\CR_{\rd}<\kCR_{\ind}$.
We consider the game $G$ with $R=\{0,1\}$ and $S=\{(0,0,\dots,0,1)\}$.
Then, for any random order distribution, the request sequence belong to $S$ with probability at most $1/n$.
Thus, by predicting that the request sequence does not belong to $S$, one can get an expected profit $1-1/n$.
On the other hand, when the input distribution $D\in\Delta_{\ind}$ satisfies $D(\bm{r})=\prod_{i=1}^n D_i(r_i)$
where $D_1(0)=\cdots=D_{n-1}(0)=1$ and $D_n(0)=D_n(1)=1/2$,
we cannot predict that the request sequence belongs to $S$ or not with probability better than $1/2$.
Thus, $\CR_{\rd}=n/(n-1)$ and $\kCR_{\ind}=2$ and hence the inequality holds.
\end{proof}

\bibliographystyle{abbrv}
\bibliography{rand_online}

\appendix
\section{Omitted Proof}
\begin{lemma}\label{lemma:odds_iid}
No online algorithm can achieve a competitive ratio better than $e$ even if the input is a known i.i.d.\ distribution.
\end{lemma}
\begin{proof}
Let us assume that $\Pr[X_i=1]=1/\sqrt{n}$ and $\Pr[X_i=0]=1-1/\sqrt{n}$ for each $i=1,\dots,n$.
Consider the following linear programming (LP):
\begin{align*}
\begin{array}{rll}
\max       \quad&\sum_{i=1}^{n} {\left(1-\frac{1}{\sqrt{n}}\right)}^{n-i}q_{i}&\\
\text{s.t.}\quad&\sqrt{n}\cdot q_{i}\le 1-\sum_{j=1}^{i-1}q_j&(i\in[n]),\\
                &q_{i}\ge 0&(i\in[n])\\
\end{array}
\end{align*}
where \([n]=\{1,2,\dots,n\}\).

We claim that the optimal value for the LP gives
an upper bound of the probability of win for the problem.
Let $q_{i}$ be the probability of selecting the $i$th random variable as success.
Then the probability of win is
\begin{align*}
\sum_{i=1}^{n} {\left(1-\frac{1}{\sqrt{n}}\right)}^{n-i}q_{i}
\end{align*}
because if $X_{i+1}=\cdots=X_n=0$ with probability ${\left(1-\frac{1}{\sqrt{n}}\right)}^{n-i}$.
Also, $p_{ij}$ must satisfy the following relation:
\begin{align*}
  p_{i}
  &\le \Pr[\text{$1,\dots,i-1$ are not selected}]\cdot\Pr[X_i=1]
  =\left(1-\sum_{j=1}^{i-1}q_j\right)\cdot\frac{1}{\sqrt{n}}.
\end{align*}
Thus, the LP present an upper bound of the success probability.

To evaluate the optimal value, now we consider the dual problem:
\begin{align*}
\begin{array}{rll}
\min       \quad&\sum_{i=1}^{n} r_{i}&\\
\text{s.t.}\quad&\sqrt{n}\cdot r_i+\sum_{j=i+1}^n r_j\ge {\left(1-\frac{1}{\sqrt{n}}\right)}^{n-i}&(i\in[n]),\\
                &r_{i}\ge 0&(i\in[n]).
\end{array}
\end{align*}
Let $r_i^*=\max\{\frac{1}{\sqrt{n}}\cdot{(1-1/\sqrt{n})}^{n-i}-\frac{n-i}{n}\cdot{(1-1/\sqrt{n})}^{n-i-1},0\}$.
Then it is easy to check that $r_i^*$ is a feasible solution for the dual LP\@.

Thus, the optimal value of the LP (and the dual LP) is upper bounded by
\begin{align*}
  \sum_{i=1}^{n} r_{i}^*
  &=\sum_{i=1}^n \max\left\{\frac{1}{\sqrt{n}}\cdot{(1-1/\sqrt{n})}^{n-i}-\frac{n-i}{n}\cdot{(1-1/\sqrt{n})}^{n-i-1},0\right\}\\
  &=\sum_{i=n-\lfloor\sqrt{n}\rfloor+1}^n \left(\frac{1}{\sqrt{n}}\cdot{(1-1/\sqrt{n})}^{n-i}-\frac{n-i}{n}\cdot{(1-1/\sqrt{n})}^{n-i-1}\right)\\
  &= \frac{1}{\sqrt{n}}\cdot\frac{1-{(1-1/\sqrt{n})}^{\lfloor\sqrt{n}\rfloor}}{1-(1-1/\sqrt{n})}-\frac{1}{n}\cdot\frac{1-{(1-1/\sqrt{n})}^{\lfloor\sqrt{n}\rfloor}-{(1-1/\sqrt{n})}^{\lfloor\sqrt{n}\rfloor}}{1/n}\\
  &={(1-1/\sqrt{n})}^{\lfloor\sqrt{n}\rfloor}\to\frac{1}{e} \quad(n\to\infty).
\end{align*}
Hence, $1/e$ is an upper bound of the probability of win for the problem.

Moreover, since $\Pr[X_1=\cdots=X_n=0]\to 0$ as $n\to\infty$,
there exists no online algorithm that can achieve a competitive ratio better than $e$.
\end{proof}

\end{document}